\documentclass{article}

\usepackage{textcomp}
\usepackage{verbatim}
\usepackage{amsthm}
\usepackage{amssymb}
\usepackage{amsmath}
\usepackage{enumerate}

\newtheorem{theorem}{Theorem}
\newtheorem{lemma}[theorem]{Lemma}
\newtheorem{definition}{Definition}
\newtheorem{corollary}{Corollary}

\usepackage{graphicx}
\usepackage{grffile}
\date{}

\title{On Strong Centerpoints}
\author{Pradeesha Ashok \thanks{Department of Computer Science and Automation,
				Indian Institute of Science, Bangalore, India.
				Email :\texttt{pradeesha@csa.iisc.ernet.in}}
\and Sathish Govindarajan\thanks{Department of Computer Science and Automation,
				Indian Institute of Science, Bangalore, India.
				Email :\texttt{gsat@csa.iisc.ernet.in}}}
\begin{document}
%
%
%
\maketitle
\begin{abstract}
Let $P$ be a set of $n$ points in $\mathbb{R}^d$ and $\mathcal{F}$ be a family of geometric objects. We call a point $x \in P$ a strong centerpoint of $P$ w.r.t $\mathcal{F}$ if $x$ is contained in all $F \in \mathcal{F}$ that contains more than $cn$ points from $P$, where $c$ is a fixed constant. A strong centerpoint does not exist even when $\mathcal{F}$ is the family of halfspaces in the plane. We prove the existence of strong centerpoints with exact constants for convex polytopes defined by a fixed set of orientations. We also prove the existence of strong centerpoints for abstract set systems with bounded intersection.

 \end{abstract}


 \section{Introduction}

Let $P$ be a set of $n$ points in $\mathbb{R}^d$. A point $x \in \mathbb{R}^d$ is said to be a \emph{centerpoint} of $P$ if any halfspace that contains $x$ contains at least $n\over{d+1}$ points of $P$. Equivalently, $x$ is a centerpoint if and only if $x$ is contained in every convex object that contains more than $\frac{d}{d+1} n$ points of $P$. It has been proved that a centerpoint exists for any pointset $P$ and the constant $\frac{d}{d+1}$ is tight~\cite{Rad46}.

The notion of centerpoint has found many applications in statistics, combinatorial geometry, geometric algorithms, etc\cite{BF84, MTT93, MTT97, Yao83}. Linear time algorithms to compute approximate centerpoint is given in \cite{CEM93, Mat91, Meg85, MS10,Ver97}. Jadhav and Mukhopadhyay\cite{JM94} gave a linear time algorithm to compute a centerpoint in the plane. Chan\cite{Cha04} gave a randomized algorithm that compute the centerpoint in $\mathbb{R}^d$ in $O(n^{d-1})$ time.

The centerpoint question i.e., finding a constant $\epsilon, 0 \leq \epsilon \leq 1$, such that there exists a centerpoint for any pointset that is contained in all objects of a certain type that contains more than $\epsilon$ fraction of the points, has been asked for special classes of convex objects. Aronov et al.\cite{AAH09} proved tight bounds for centerpoint for the family of halfplanes, axis-parallel rectangles and disks in $\mathbb{R}^2$. Another well-studied generalization of centerpoint is to allow more than one point. This is related to an area called $\epsilon$-nets. 

$N \subset P$ is said to be a (strong) $\epsilon$-net of $P$ w.r.t a family of geometric objects $\mathcal{R}$ if $N \cap R \ne \emptyset$ for all $R \in \mathcal{R}$ that contains more than $\epsilon n$ points from $P$. $N$ is called a weak $\epsilon$-net if $N$ is not restricted to be a subset of $P$ but is allowed to be any subset of $\mathbb{R}^d$. Haussler and Welzl\cite{HW87} showed that small-sized $\epsilon$-nets exist for range spaces of bounded VC-dimension. Small $\epsilon$-net question investigates the bounds on $\epsilon$ when the size of $\epsilon$-net is fixed as a small constant\cite{AAH09, BZ06,Dul06, MR09}. Note that a centerpoint is a weak $\epsilon$-net of size one, w.r.t convex objects. 


In general, a centerpoint need not be a point of $P$ and can be any point in $\mathbb{R}^d$. In this paper, we study the question of enforcing the centerpoint to be a point of $P$. We call such a centerpoint a strong centerpoint. 

We now define strong centerpoints in an abstract setting.
\begin{definition}
 Let $P$ be a set of $n$ elements and $\mathcal{S}$ be a family of subsets of $P$. Then $p \in P$ is called the \emph{strong centerpoint} of $P$ w.r.t $\mathcal{S}$ if $p \in S $ for all $S \in \mathcal{S}$ such that $\vert S  \vert > c n$, where $0 <c < 1$ is a fixed constant. 
\end{definition}

It is easy to see that a strong centerpoint does not exist even when $P$ is a set of $n$ points in $\mathbb{R}^d$ and $\mathcal{S}$ is defined by halfspaces. Let $P$ be a set of $n$ points in convex position. For any point $p \in P$, there exists a halfspace that contains all the points in $P \setminus \{p\}$. Therefore, a strong centerpoint does not exist for halfspaces, and therefore, for disks and convex objects. Ashok et al.\cite{AGK10} proved the existence of strong centerpoints for axis-parallel rectangles in $\mathbb{R}^2$. To the best of our knowledge, no other results on strong centerpoints are known.

In this paper, we study the strong centerpoint question and prove tight bounds for some classes of geometric and abstract objects. 

\subsection{Our results}
Let $P$ be a set of $n$ points in $\mathbb{R}^d$.
\begin{enumerate}
 \item We prove a strong centerpoint exists for a special class of convex polytopes viz. convex polytopes defined by a set of fixed orientations. Let $\mathcal{F}$ represent the set of convex polytopes defined by a set of $k$ fixed orientations. Then there exists a strong centerpoint $p \in P$ such that $p$ is contained in all $F \in \mathcal{F}$ that contains more than $(1-\frac{1}{k})n$ points from $P$. Moreover, this bound is tight. Our proof is constructive and can be converted into a linear time algorithm to compute such a strong centerpoint. Our argument is a generalization of a construction given in Lemma 2 of \cite{AGK10}.

%
 \item We prove the existence of a strong centerpoint for set systems with ``bounded intersection''. Let $(P,\mathcal{S}_k)$ be a set system where $P$ is a set of $n$ elements and $\mathcal{S}_k$ is a collection of subsets of $P$ with the property that the intersection of any $k$ subsets in $\mathcal{S}_k$ is either equal to the intersection of strictly fewer sets among them or contains at most one element of $P$. We prove that a strong centerpoint $p \in P$ exists such that $p$ is contained in all $S \in \mathcal{S}_k$ such that $\vert S \cap P \vert > (1-\frac{1}{k})n$.
\end{enumerate}
Section \ref{def} gives some definitions and preliminary results that will be used in subsequent sections. Section \ref{convex} proves the existence of strong centerpoints and gives tight bounds for the family of convex polytopes defined by a set of fixed orientations. 
In section \ref{abstract_new}, we prove the existence of strong centerpoint for set systems with bounded intersection.
 \section{Definitions and Preliminary Results}\label{def}
 In this section, we give some definitions and preliminary results that will be used in subsequent sections.

\begin{definition}
 The \emph{orientation} of a halfspace is the direction of the outward normal to that halfspace.
\end{definition}
Note that if two halfspaces $H_1$ and $H_2$ are of the same orientation then one of them is contained in the other. 
\begin{definition}
Let $C$ be a convex polytope in $\mathbb{R}^d$. Let $H_1,\cdots,H_k$ be the halfspaces defined by faces of $C$ such that $C=\bigcap \limits_{i=1}^kH_i$. We call $H_1,\cdots,H_k$ the defining halfspaces of $C$.
\end{definition}

For a general convex polytope $C$, the defining halfspaces of $C$ can be of any orientation. We consider a class of convex polytopes where the orientation of the defining halfspaces belong to a fixed set.

\noindent Let $\mathcal{O}$ be a set of orientations.
\begin{definition}
 A family of convex polytopes $\mathcal{C}$ is said to be \emph{defined by $\mathcal{O}$} if for any $C \in \mathcal{C}$, the orientations of all the defining halfspaces of $C$ belong to $\mathcal{O}$.
\end{definition} 

Many common classes of geometric objects fall into this category. For example, axis-parallel boxes are defined by a set of $2d$ fixed orientations viz. direction of positive and negative axes in all the $d$ dimensions. Some other geometric objects that fall into this category are
\begin{figure}
 \begin{center}
%
%
\includegraphics[scale=0.5]{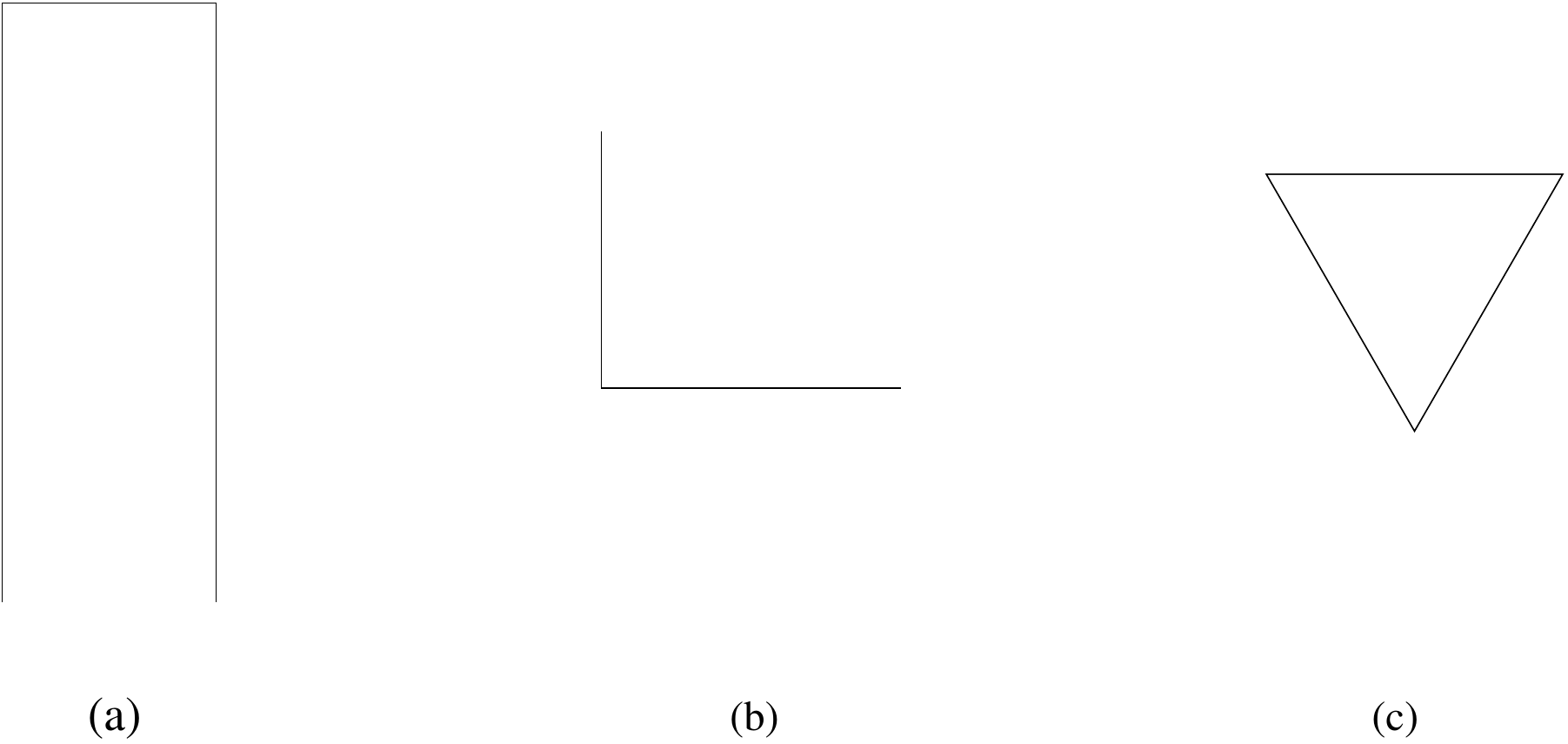} 
 \caption{Examples of Convex Polytopes defined by a fixed set of Orientations: (a) Skylines in $\mathbb{R}^2$  (b) Orthants of fixed orientations in $\mathbb{R}^2$  (c) Downward facing equilateral triangles}
  \label{convex_eg}
 \end{center}

\end{figure}

\begin{itemize}
 \item Skylines in $\mathbb{R}^d$: Skylines are axis-parallel boxes where one fixed axis, say the negative axis in the $d^{th}$ dimension, is unbounded. These are convex polytopes defined by $2d-1$ fixed orientations viz., direction of positive axes in all the $d$ dimensions and direction of negative axes in all $d$ dimensions except the last one. Figure \ref{convex_eg}(a) shows skylines in $\mathbb{R}^2$. 
 \item Orthants of fixed orientation in $\mathbb{R}^d$: Orthants in $\mathbb{R}^d$ can be considered as intersection of $d$ mutually orthogonal halfspaces. For orthants of fixed orientation, the orientations of the defining halfspaces are fixed, say as the direction of positive axes in all the $d$ dimensions. Figure \ref{convex_eg}(b) shows orthants in $\mathbb{R}^2$.
\item Downward facing equilateral triangles: These are equilateral triangles with one side parallel to the X-axis and the corner opposite to this side lying below it\cite{Pan11}. Therefore, this family of triangles are defined by a set of three fixed orientations viz., directions at $90 ^\circ, 210^\circ$ and $330^\circ$ with positive $X$ axis(See figure \ref{convex_eg}(c)). 
 \item Homothets of a $k$-faced convex polytope where $k$ is a fixed constant.
\end{itemize}


\begin{definition}
 A set system $(P,\mathcal{S}_k)$ is said to be a \emph{set system of bounded intersection} if it satisfies the following property: For any $k$ sets in $\mathcal{S}_k$, their intersection is either equal to the intersection of strictly fewer sets among them or contains atmost one element of $P$. 
\end{definition}
Many geometric set systems are set systems with bounded intersection. For example, geometric set systems defined by hyperplanes in $\mathbb{R}^d$ have the property that the intersection of any $d$ sets is either equal to the intersection of strictly fewer sets among them or contains atmost one point. Similarly,  geometric set systems defined by straight lines have the property that any two sets intersect in atmost one point and set systems defined by circles have the property that any three sets intersect in atmost one point. 
\begin{lemma}\label{lines}
 Let $P$ be a set of $n$ elements and $\mathcal{S}_2$ be a collection of subsets of $P$ with the property that the intersection of any two sets in $\mathcal{S}_2$ contains atmost one element of $P$. Then there exists $p\in P$ such that $p$ is contained in all $S \in \mathcal{S}_2$ that contains more than $\frac{n}{2}$ elements from $P$. 
\end{lemma}

\begin{proof}
If none of the sets in $\mathcal{S}_2$ contain more than $\frac{n}{2}$ elements from $P$ then there is nothing to prove. Therefore, assume such a set $S$ exists.
Let $p$ be any element in $ S$. We claim that any set that contains more than $\frac{n}{2}$ elements from $P$ contains $p$. Let $S_1$ be any nonempty set in $\mathcal{S}_2$. If $S_1=S$ then $S_1$ contains $p$. Otherwise $\vert S \cap S_1 \vert \leq 1$ and $\vert S_1 \vert \leq \frac{n}{2}$.
\end{proof}

\section{Convex Polytopes defined by a fixed set of Orientations}\label{convex}
Let $P$ be a set of $n$ points in $\mathbb{R}^d$. Let $\mathcal{C}$ be a family of convex polytopes in $\mathbb{R}^d$ defined by a fixed set of orientations. We show the existence of strong centerpoints for $P$ w.r.t $\mathcal{C}$ and prove tight bounds. 
\begin{theorem} \label{convexproof}
 Let $\mathcal{C}$ be a family of convex polytopes in $\mathbb{R}^d$ defined by $\mathcal{O}$ and $\vert \mathcal{O}\vert =k$.  Then there exists a strong centerpoint $p \in P$ with respect to $\mathcal{C}$ such that $p$ is contained in all $C \in \mathcal{C}$ that contains more than $(1-\frac{1}{k})n$ points from $P$. Moreover this bound is tight.
\end{theorem}
\begin{proof}

For each orientation $\hat{i}$ in $\mathcal{O}$, let halfspace $H_i$ have orientation $\hat{i}$ and  $\vert H_i \cap P \vert= n-\frac{n}{k}+1$. Let $E$ represent the region $\bigcap \limits_{i=1}^k H_i$. We claim that $\vert E \cap P \vert \ne 0$.

Let $\overline{H_i}$ represent $\mathbb{R}^d\setminus H_i$ for all $i$, $1\leq i \leq k$. Now $E=\mathbb{R}^d\setminus\{\bigcup \limits_{i=1}^k \overline{H_i}\}$
\begin{eqnarray*}
\vert E \cap P \vert &=&n-\vert (\bigcup\limits_{i=1}^k \overline{H_i}) \cap P \vert \\
&\geq& n-k(\frac{n}{k}-1)\\
&\geq& k
\end{eqnarray*}
Therefore, region $E$ contains at least $k$ points from $P$. 

Let $p \in P$ be any point in $E$. We claim that $p$ is a strong centerpoint for $P$ w.r.t $\mathcal{C}$ i.e., $p$ is contained in all $C \in \mathcal{C}$ that contains more than $n-\frac{n}{k}+1$ points from $P$.

 Let $C \in \mathcal{C}$ and $\vert C \cap P \vert > (1-\frac{1}{k})n$. Let $H^\prime_1,H^\prime_2,\cdots,H^\prime_k$ be the defining halfspaces of $C$. W.l.o.g assume that $H_i$ and $H^\prime_i$ have the same orientation for all $i$, $1\leq i \leq k$. Now, either $H_i^\prime \subset H_i$ or $H_i \subseteq H_i^\prime$. Suppose $H_i^\prime \subset H_i$ for some $i$, $1\leq i \leq k$. This implies that $C \cap \overline{H_i}=\emptyset$. Since $\overline{H_i}$ contains $\frac{n}{k}$ points of $P$ this implies that $\vert C \cap P \vert \leq (1-\frac{1}{k})n$, a contradiction. Therefore, $H_i \subseteq H_i^\prime$ for all $i, 1 \leq i \leq k$. Therefore, 
 

 \begin{center}
  $ E=\bigcap\limits_{i=1}^k H_i \subseteq \bigcap\limits_{i=1}^k H^\prime_i=C$
 \end{center}

\noindent Since $C$ contains region $E$, $C$ contains $p$.
%

To prove the lower bound, let $P$ be arranged as $k$ subsets $P_1,P_2,\cdots,P_k$ of equal size. Each $P_i$ is placed at unit distance from the origin along the orientation $\hat{i}\in \mathcal{O}$. 

Therefore, for all $i$, $1\leq i \leq k$, there exists halfspaces $H_i$ and $H^\prime_i$ of orientation $\hat{i}$ such that $H_i\cap P =P \setminus P_i$   and $H_i^\prime \cap P=P$. For any point $p\in P_i$, $C =\left( \bigcap\limits_{j=1}^k H^\prime_j \right) \cap H_i$ contains $(1-\frac{1}{k})n$ points from $P$ but does not contain $p$.
\end{proof}
\begin{corollary}
 Let $P$ be a set of $n$ points in $\mathbb{R}^2$ and $\mathcal{T}$ represent the family of downward facing equilateral triangles. A strong centerpoint $p\in P$ exists w.r.t $\mathcal{T}$  such that $p$ is contained in all $T \in \mathcal{T}$ that contain more than $\frac{2n}{3}$ points from $P$.
\end{corollary}
\begin{proof}
 The result follows from Theorem \ref{convexproof} and the fact that $\mathcal{T}$ is defined by a set of three fixed orientations.
\end{proof}
\begin{corollary}
 Let $P$ be a set of $n$ points in $\mathbb{R}^d$ and $\mathcal{K}$ represent the family of skylines. A strong centerpoint $p\in P$ exists w.r.t $\mathcal{K}$  such that $p$ is contained in all $K \in \mathcal{K}$ that contain more than $(1-\frac{1}{2d-1})n$ points from $P$.
\end{corollary}
\begin{proof}
The result follows from Theorem \ref{convexproof} and the fact that $\mathcal{K}$ is defined by a set of $2d-1$ fixed orientations.
\end{proof}

\begin{corollary}
 Let $P$ be a set of $n$ points in $\mathbb{R}^d$ and $\mathcal{T}$ represent the family of orthants of fixed orientation. A strong centerpoint $p\in P$ exists w.r.t $\mathcal{T}$  such that $p$ is contained in all $T \in \mathcal{T}$ that contain more than $(1-\frac{1}{d})n$ points from $P$.
\end{corollary}
\begin{proof}
 The result follows from Theorem \ref{convexproof} and the fact that $\mathcal{T}$ is defined by a set of $d$ fixed orientations.
\end{proof}
\begin{corollary}
Let $P$ be a set of $n$ points in $\mathbb{R}^d$ and $\mathcal{M}$ represent the family of homothets of a $k$-faced convex polytope, $C$. A strong centerpoint $p\in P$ exists w.r.t $\mathcal{M}$  such that $p$ is contained in all $M \in \mathcal{M}$ that contain more than $(1-\frac{1}{k})n$ points from $P$.
\end{corollary}
\begin{proof}
 The result follows from Theorem \ref{convexproof} and the fact that $\mathcal{M}$ is defined by a set of $k$ fixed orientations, viz. the orientations of the faces of $C$.
\end{proof}

\section{Set Systems with Bounded Intersection}\label{abstract_new}

Let $(P,\mathcal{S}_k)$ be a set system with the following property : the intersection of any $k$ sets in $\mathcal{S}$ is either equal to the intersection of strictly fewer sets among them or has size at most one.  We prove that a strong centerpoint exists for this set system and prove tight bounds.
\begin{theorem}
 Let $P$ be a set of $n$ elements. Then there exists a strong centerpoint $p \in P$ such that $p$ is contained in any $S \in \mathcal{S}_k$ that contains more than $(1-\frac{1}{k})n$ elements from $P$.
\end{theorem}
\begin{proof}
 We prove the result by induction on $k$. When $k=2$, the result holds by Lemma \ref{lines}.
 
 We now prove the result for a general $k$.
 Let $S^\prime \in \mathcal{S}_k$ be such that $\vert S^\prime \vert > (1-\frac{1}{k})n$. If no such $S^\prime$ exists then any element in $P$ is a strong centerpoint. Let $\vert S^\prime \vert=n^\prime$. Let us define a new set system $(S^\prime, \mathcal{S}_{k-1})$ where $\mathcal{S}_{k-1}=\{ S_1\cap S^\prime | S_1 \in \mathcal{S}_k\}$. Now $(S^\prime, \mathcal{S}_{k-1})$ has the property that the intersection of any $k-1$ sets in $\mathcal{S}_{k-1}$ either contains atmost one element of $S^\prime$ or is equal to the intersection of strictly fewer sets among them. Let $S_1,\dots,S_{k-1}$ be any $k-1$ sets in $\mathcal{S}_k$ such that $S_i \cap S^\prime \ne \emptyset$ for $1 \leq i \leq k-1$. Therefore, $S_i \cap S^\prime, 1\leq i \leq k-1$ are sets in $S_{k-1}$. If $\vert \left(\bigcap\limits_{1\leq i \leq k-1} S_i \right)\cap S^\prime \vert =1$, then it is easy to see that $\vert \bigcap\limits_{1\leq i \leq k-1} \left(S_i \cap S^\prime\right) \vert =1$. Similarly, if $\left(\bigcap\limits_{1\leq i \leq k-1} S_i \right)\cap S^\prime$ has the property that it is also the intersection of strictly fewer sets among them then the intersection $\bigcap\limits_{1\leq i \leq k-1} \left(S_i \cap S^\prime\right)$ has the same property.
 
 Let $p$ be a strong centerpoint of $S^\prime$ w.r.t $\mathcal{S}_{k-1}$ such that $p$ is contained in all $S_1 \in \mathcal{S}_{k-1}$ that contains more than $(1-\frac{1}{k-1})n^\prime$ elements from $S^\prime$. By inductive hypothesis, such a strong centerpoint exists. We claim that $p$ is contained in all $S \in \mathcal{S}_k$ that contains more than $(1-\frac{1}{k})n$ elements.

Assume there exists $S_1 \in \mathcal{S}_k$ that does not contain $p$. We prove that $S_1$ contains atmost $(1-\frac{1}{k})n$ elements from $P$. If $S_1 \cap S^\prime =\emptyset$ then $\vert S_1 \vert \leq (1-\frac{1}{k})n$. Therefore assume that $S_1 \cap S^\prime \ne \emptyset$. Since $S_1$ does not contain $p$, $S_1$ contains atmost $(1-\frac{1}{k-1})n^\prime$ points from $S^\prime$. Hence,
\begin{eqnarray*}
 \vert S_1 \cap P \vert &\leq& (1-\frac{1}{k-1})n^\prime+n-n^\prime \\
&\leq& n-\frac{n^\prime}{k-1}\\
&\leq& (1-\frac{1}{k})n 
\end{eqnarray*}
(since $n^\prime \geq (1-\frac{1}{k})n$)
 \begin{corollary}
 Let $P$ be a set of $n$ points in $\mathbb{R}^d$, $d \geq 2$, and $\mathcal{H}_d$ represent the family of all hyperplanes in $\mathbb{R}^d$. Then there exists a strong centerpoint $p\in P$ w.r.t $\mathcal{H}_d$ such that $p$ is contained in all hyperplanes that contain more than $(1-\frac{1}{d})n$ points from $P$. Moreover, this bound is tight.
 \end{corollary}

\end{proof}

\section{Conclusion and Open Questions}
We investigated the existence of strong centerpoints and proved tight bounds for convex polytopes defined by a fixed set of orientations and hyperplanes in $\mathbb{R}^d$. We also proved bounds for strong centerpoints for an abstract set system with bounded intersection. It will be interesting to see if there are other classes of geometric objects for which a strong centerpoint exists.

\end{document}